\documentclass[a4paper,UKenglish,cleveref, autoref, thm-restate,authorcolumns]{lipics-v2019}

\usepackage{cite}
\usepackage[noend]{algorithmic}
\usepackage{xspace}
\usepackage{caption}
\usepackage{stackengine}
\usepackage{algorithm}
\usepackage{paralist}
\usepackage{color,colortbl}
\usepackage{subcaption}
\usepackage{paralist}

\usepackage{ifmtarg}

\def\e#1{\emph{#1}}

\newcommand{\eat}[1]{}

\newcommand{\set}[1]{\{#1\}}
\def\eqdef{\stackrel{\textsf{\tiny def}}{=}}

\newcommand{\algname}[1]{{\sf #1}}
\def\myrulewidth{3.20in}

\def\therule{\rule{\myrulewidth}{0.2pt}}

\newenvironment{insidecode}[3]
{
\begin{tabular}{p{\myrulewidth}}
\multicolumn{1}{c}{\rule{0mm}{3mm}{\bf #3} $\algname{#1}(\mbox{#2})$\vspace{-0.6em}}\\
\therule\vskip-0.8em\therule
\vspace{0em}
\begin{algorithmic}[1]}
{\end{algorithmic}
\vskip-0.3em\therule
\end{tabular}}

\newcommand{\depset}{\mathrm{\Delta}}

\def\phi{\varphi}

\newtheorem{problem}[theorem]{Problem}

\newenvironment{repeatresult}[2]
{\vskip0.5em\par\textsc{#1} #2.\em}
{\vskip1em}

\newenvironment{reptheorem}[1]{\begin{repeatresult}{Theorem}{#1}}{\end{repeatresult}}
\newenvironment{replemma}[1]{\begin{repeatresult}{Lemma}{#1}}{\end{repeatresult}}

\def\N{\mathcal{N}}

\newtheorem{commentthm}[theorem]{Comment}

\newenvironment{citedtheorem}[1]
{\begin{theorem}\hskip-0.2em\e{\cite{#1}}\,\,}
{\end{theorem}}

\def\cost{\mathit{cost}}
\def\top{\mathsf{top}}

\def\val#1{\texttt{#1}}

\def\rel#1{\textsc{#1}}
\def\att#1{\textsf{#1}}

\newenvironment{citemize}
{\begin{compactitem}}
{\end{compactitem}}

\DeclareMathOperator*{\argmin}{argmin}

\def\ra{\rightarrow}

\newenvironment{subroutine}
{\begin{algorithm}\floatname{algorithm}{Subroutine}}
{\end{algorithm}}

\newcounter{subroutine}

\def\cost{\mathsf{cost}}

\def\dbcost{\mathsf{cost}}
\def\vio{\mathsf{vio}}

\def\assn{\mathrel{{:}{=}}}

\title{Database Repairing with Soft Functional Dependencies} 

\author{Nofar Carmeli}{Technion, Haifa,  Israel}{}{}{}
\author{Martin Grohe}{RWTH Aachen University, Germany}{}{}{}
\author{Benny Kimelfeld}{Technion, Haifa, Israel}{}{}{}
\author{Ester Livshits}{Technion, Haifa, Israel}{}{}{}
\author{Muhammad Tibi}{Technion, Haifa, Israel}{}{}{}

\authorrunning{N. Carmeli, M. Grohe, B. Kimelfeld, E. Livshits, and M. Tibi} 

\Copyright{Nofar Carmeli, Martin Grohe, Benny Kimelfeld, Ester Livshits, and Muhammad Tibi} 



\nolinenumbers 

\newcommand{\nofar}[1]{\textcolor{purple}{\textbf{Nofar:} #1}}

\definecolor{darkgreen}{RGB}{0,143,80}

\newcommand{\martindone}[1]{}

\begin{document}

\maketitle

\begin{abstract}
  A common interpretation of soft constraints penalizes the database
  for every violation of every constraint, where the penalty is the
  cost (weight) of the constraint. A computational challenge is that
  of finding an optimal subset: a collection of database tuples that
  minimizes the total penalty when each tuple has a cost of being
  excluded. When the constraints are strict (i.e., have an infinite
  cost), this subset is a ``cardinality repair'' of an inconsistent
  database; in soft interpretations, this subset corresponds to a
  ``most probable world'' of a probabilistic database, a ``most likely
  intention'' of a probabilistic unclean database, and so on.  Within the class of
  functional dependencies, the complexity of finding a cardinality
  repair is thoroughly understood. Yet, very little is known about the
  complexity of this problem in the more general soft semantics. This
  paper makes a significant progress in this direction. In addition to
  general insights about the hardness and approximability of the
  problem, we present algorithms for two special cases: a single
  functional dependency, and a bipartite matching. The latter is the
  problem of finding an optimal ``almost matching'' of a bipartite
  graph where a penalty is paid for every lost edge and every
  violation of monogamy.
\end{abstract}

\section{Introduction}

Soft variants of database constraints (also referred to as \e{weak} or \e{approximate} constraints) have been a building block of various challenges in data management. In constraint discovery and mining, for instance, the goal is to find constraints, such as Functional Dependencies (FDs)~\cite{DBLP:journals/cj/HuhtalaKPT99,DBLP:journals/cbm/CombiMSSAMP15,DBLP:conf/apweb/LiLCJY16} and beyond~\cite{DBLP:journals/pvldb/ChuIP13,DBLP:journals/pvldb/LivshitsHIK20,DBLP:journals/pvldb/PenaAN19}, that generally hold in the database but not necessarily in a perfect manner.  There, the reason for the violations might be rare events (e.g., agreement on the zip code but not the state) or noise (e.g., mistyping).
Soft constraints also arise when reasoning about uncertain data~\cite{DBLP:conf/icdt/SaIKRR19, DBLP:conf/pods/JhaRS08, DBLP:journals/vldb/SenDG09,GVSBUDA14}---the database is viewed as a probabilistic space over possible worlds, and the violation of a weak constraint in a possible world is viewed as evidence that affects the world's probability.

Our investigation concerns the latter application of soft constraints.  To be more precise, the semantics is that of a \e{parametric factor graph}: the probability of a possible world is the product of \e{factors} where every violation of the constraint contributes one factor; in turn, this factor is a weight that is assigned upfront to the constraint. This approach is highly inspired by successful concepts such as the Markov Logic Network (MLN)~\cite{Richardson:2006:MLN:1113907.1113910}.
The computational challenges are the typical ones of probabilistic modeling: marginal inference (compute the probability of a query answer) and maximum likelihood (find the most probable world)---the problem that we focus on here.

More specifically, we investigate the complexity of finding a most probable world in the case where the constraints are FDs. By taking the logarithms of the factors,
this problem can be formally defined as follows. We are given a database $D$ and a set $\Delta$ of FDs, where every tuple and every FD has a weight (a nonnegative number). We would like to obtain a cleaner subset $E$ of $D$ by deleting tuples. The cost of $E$ includes a penalty for every deleted tuple and a penalty for every violation of (i.e., pair of tuples that violates) an FD; the penalties are the weights of the tuple and the FD, respectively. The goal is to find a subset $E$ with a minimal cost. In what follows, we refer to such $E$ as an \e{optimal subset} and to the optimization problem of finding an optimal subset as \e{soft repairing}. The optimal subset corresponds to the ``most likely intention'' in the Probabilistic Unclean Database (PUD) framework of De Sa, Ilyas, Kimelfeld, R{\'{e}} and Rekatsinas~\cite{DBLP:conf/icdt/SaIKRR19} in a restricted case that is studied in their work, and to the ``most probable world'' in the probabilistic database model of Sen, Deshpande and Getoor~\cite{DBLP:journals/vldb/SenDG09}. In the special case where the FDs are hard constraints (i.e., their weight is infinite or just too large to pay), an optimal subset is simply what is known as a ``cardinality repair''~\cite{DBLP:conf/icdt/LopatenkoB07} or, equivalently~\cite{DBLP:journals/tods/LivshitsKR20}, a ``most probable database''~\cite{GVSBUDA14}.

The computational challenge of soft repairing is that there are exponentially many candidate subsets.
We investigate the data complexity of the problem, where the database schema and the FD set are fixed, and the input consists of the database $D$ and all involved weights. Moreover, we assume that $D$ consists of a single relation; this is done without loss of generality, since the problem
boils down to soft repairing each relation independently (since an FD does not involve more than one relation).
 
The complexity of the problem is very well understood in the case of hard constraints (cardinality repairs).  Gribkoff, Van den Broeck and Suciu~\cite{GVSBUDA14} established complexity results for the case of unary FDs (having a single attribute on the left-hand side), and Livshits, Kimelfeld and Roy~\cite{DBLP:journals/tods/LivshitsKR20} completed the picture to a full (effective) dichotomy over all possible sets of FDs. For example, the problem is solvable in polynomial time for the FD sets $\set{A\ra B}$, $\set{A \ra B,B\ra A}$ and $\set{A\ra B,B\ra A,B\ra C}$, but is NP-hard for $\set{A \ra B,B\ra C}$. In contrast, very little is known about the more general case where the FDs are soft (and violations are allowed), where the problem seems to be fundamentally harder, both to solve and to reason about. Clearly, for every $\Delta$ where  it is intractable to find a cardinality repair, the soft version is also intractable. But the other direction is false (under conventional complexity assumptions). For example, soft repairing is hard for
$\Delta=\set{A\ra B,B\ra A,B\ra C}$, for the following reason. We can set the weights of $A\ra B$ and $B\ra C$ to be very high, making each of them a hard constraint in effect, and the weight of $B\ra A$ very low, making it ignorable in effect. Hence, an optimal subset is a cardinality repair for $\set{A\ra B,B\ra C}$ that, as said above, is hard to compute.

So, which sets of FDs have a tractable soft repairing?
The only polynomial-time algorithm we are aware of is that of De Sa et al.~\cite{DBLP:conf/icdt/SaIKRR19} for the special case of a single key constraint, that is, $\Delta=\set{X\ra Y}$ where $XY$ contain all of the schema attributes;
they have left the more general case (that we study here) open.
In this work, we make substantial progress
in answering this question by presenting algorithms for two types of FD sets: \e{(a)} a single FD \e{and (b)} a matching constraint.

The first type generalizes the tractability of De Sa et al.~\cite{DBLP:conf/icdt/SaIKRR19} from a key constraint to an arbitrary FD (as long as it is the only FD in $\Delta$). Like theirs, our algorithm employs dynamic programming, but in a more involved fashion. This is because their algorithm is based on the fact that in a key constraint $X\ra Y$,
any two tuples that agree on $X$ are necessarily conflicting.
We also show that our algorithm can be generalized to additional sets of FDs. For example, it turns our that
the FD set
$\set{\textrm{name}\ra\textrm{address}\,,\,\textrm{name address}\ra\textrm{email}}$ is tractable as well. (Note that the $\textrm{address}$ attribute on the left-hand side of the second FD is not redundant, as in the ordinary semantics, since the FDs are treated as soft constraints.) In Section~\ref{sec:single} we phrase the more general condition that this FD set satisfies.

The second type, matching constraints, refers to FD sets $\Delta=\set{X\rightarrow Y,X'\rightarrow Y'}$ over a schema with the attributes $A_1$, \dots,$A_k$ where $X\cup Y=X'\cup Y'=X\cup X'=\set{A_1,\dots,A_k}$. The simplest example is $\set{A\ra B,B\ra A}$ over the binary schema $(A,B)$ that represents a bipartite graph, and the problem is that of finding the best ``almost matching'' of a bipartite graph where a penalty is paid for every lost edge and every
violation of monogamy. 
A more involved example is
$\set{\textrm{fn ln}\ra\textrm{addr}\,,\, \textrm{fn addr}\ra\textrm{ln}}$ over the schema
$(\textrm{fn},\textrm{ln},\textrm{addr})$. Our algorithm is based on a reduction to the
\e{Minimum Cost Maximum Flow} (MCMF) problem~\cite{10.1287/moor.15.3.430}.

Whether our algorithms cover all of tractable cases remains an open problem for future investigation. (In the Conclusions we discuss the simplest FD sets where the question is left unsolved.) We do show, however, that
there is a polynomial-time approximation algorithm with an approximation factor $3$,
that is, a subset where the penalty is at most three times the optimum.

The rest of the paper is organized as follows. We give the formal setup and the problem definition in Section~\ref{sec:preliminaries}. We then discuss the complexity of the general problem and its relationship to past results in Section~\ref{sec:complexity}. We describe our algorithm for soft repairing in Sections~\ref{sec:single} and~\ref{sec:matching} for a single FD and a matching constraint, respectively, and conclude in Section~\ref{sec:conclusions}. For lack of space, some of the proofs are given in the Appendix.

\section{Formal Setup}\label{sec:preliminaries}
We begin with preliminary definitions and terminology that we use
throughout the paper.

\subsection{Databases, FDs and Repairs}
A \e{relation schema} $R(A_1,\dots,A_k)$ consists of a relation symbol
$R$ and a set $\set{A_1,\dots,A_k}$ of attributes. A \e{database} $D$
over $R$ is a set of facts $f$ of the form $R(c_1,\dots,c_k)$, where
each $c_i$ is a \e{constant}. We denote by $f[A_i]$ the value that the
fact $f$ associates with attribute $A_i$ (i.e., $f[A_i]=c_i$).
Similarly, if $X=B_1\cdots B_k$ is a sequence of attributes from
$\set{A_1,\dots,A_k}$, then $f[X]$ is the tuple $(f[B_1],\dots,f[B_k])$.

A \e{Functional Dependency} (FD) over the relation schema
$R(A_1,\dots,A_k)$ is an expression $\varphi$ of the form
$X\rightarrow Y$ where $X,Y\subseteq\set{A_1,\dots,A_k}$. A
\e{violation} of an FD in a database $D$ is a pair $\set{f,g}$ of
tuples from $D$ that agrees on the left-hand side (i.e., $f[X]=g[X]$)
but disagrees on the right-hand side (i.e., $f[Y]\neq g[Y]$).
An FD $X\rightarrow Y$ is \e{trivial} if $Y \subseteq X$.
We denote by $\vio(D,\varphi)$ the set of all the violations of the FD
$\varphi$ in $D$. We say that $D$ \e{satisfies} $\varphi$, denoted
$D\models\varphi$, if it has no violations (i.e., $\vio(D,\varphi)$ is
empty). The database $D$ satisfies a set $\Delta$ of FDs, denoted by
$D\models\Delta$, if $D$ satisfies every FD in $\Delta$; otherwise,
$D$ violates $\Delta$ (denoted $D\not\models\Delta$).

When there is no risk of ambiguity, we may omit the specification of
the relation schema $R(A_1,\dots,A_k)$ and simply assume that the
involved databases and constraints are all over the same schema.

Let $D$ be a database and let $\Delta$ be a set of FDs. A \e{repair}
(\e{of
  $D$ w.r.t.~$\Delta$}) is a maximal consistent subset $E$; that is,
$E\subseteq D$ and $E\models\Delta$, and moreover,
$E'\not\models\Delta$ for every $E'$ such that
$E\subsetneq E'\subseteq D$. Note that the number of repairs can be
exponential in the number of facts of $D$. A \e{cardinality repair} is a
repair $E$ of a maximal cardinality (i.e., $|E|\geq |E'|$ for every repair $E'$).

\subsection{Soft Constraints}

We define the concept of \e{soft constraints} (or \e{weak constraints}
or \e{weighted rules}) in the standard way of ``penalizing'' the
database for every missing fact, on the one hand, and every violation,
on the other hand. This is the concept adopted in past work such as
the \e{parfactors} of De Sa et al.~\cite{DBLP:conf/icdt/SaIKRR19}, the
\e{soft keys} of Jha et al.~\cite{DBLP:conf/pods/JhaRS08}, and the
\e{PrDB} model of Sen et al.~\cite{DBLP:journals/vldb/SenDG09}. The
concept can be viewed as a special case of the \e{Markov Logic
  Network} (MLN)~\cite{Richardson:2006:MLN:1113907.1113910}.

Formally, let $D$ be a database and $\Delta$ a set of FDs. We assume
that every fact $f\in D$ and every FD $\varphi\in\Delta$ have a
nonnegative weight, hereafter denoted $w_f$ and $w_\varphi$,
respectively. (The weight of a fact is sometimes viewed as the log of
a validity/existence
probability~\cite{DBLP:journals/vldb/SenDG09,GVSBUDA14}.) The \e{cost}
of a subset $E$ of a database $D$ is then defined as follows.
\begin{equation}
  \label{eq:cost}
  \dbcost(E\mid D)\eqdef \left(\sum_{f\in (D\setminus E)}\!\!w_f\right)
  + \left(\sum_{\varphi\in\Delta}w_\varphi|\vio(E,\varphi)|\right)
\end{equation}

As for the computational model, we assume that every weight is a
rational number $r/q$ that is represented using the numerator and the
denominator, namely $(r,q)$, where each of the two is an integer
represented in the standard binary manner.

\subsection{Problem Definition: Soft Repairing} The problem we study
in this paper, referred to as \e{soft repairing}, is the optimization
problem of finding a database subset with a minimal cost. Since we
consider the data complexity of the problem, we associate with each
relation schema and set of FDs a separate computational problem.

\begin{problem}[Soft Repairing]
  Let $R(A_1,\dots,A_k)$ be a relation schema and $\Delta$ a
  set of
  FDs. \e{Soft repairing} (\e{for $R(A_1,\dots,A_k)$ and $\Delta$}) is
  the following optimization problem: Given a database $D$, find an
  optimal subset of $D$, that is, a subset
  $E$ of $D$ with a minimal
  $ \dbcost(E\mid D)$. \end{problem}

Note that a cardinality repair is an optimal subset in the special
case where the weight $w_\varphi$ of every FD $\varphi$ is $\infty$
(or just higher than the cost of deleting the entire database), and
the weight $w_f$ of every fact $f$ is $1$. Livshits et
al.~\cite{DBLP:journals/tods/LivshitsKR20} studied the complexity of
finding a \e{weighted cardinality repair}, which is the same as a
cardinality repair but the weight $w_f$ of every fact $f$
can be arbitrary. Hence, both types of cardinality repairs are
consistent (i.e., the constraints are strictly satisfied). In
contrast, an optimal subset in the general case may violate one or
more of the FDs. In the next section we recall the known complexity
results for cardinality and weighted cardinality repairs.

\def\auf{f^{\mathrm{a}}}
\def\instf{f^{\mathrm{i}}}
\def\pubf{f^{\mathrm{p}}}
\def\citf{f^{\mathrm{c}}}

{
\definecolor{Gray}{gray}{0.9}
\def\emprow{\multicolumn{3}{l}{}}
\begin{figure}[t]
\small
\centering
\begin{subfigure}[b]{\linewidth}\centering
\begin{tabular}{|c|c|c|c|c|c|r} 
\cline{1-6}
\rowcolor{Gray}
\multicolumn{7}{l}{\rel{Flights}}\\\cline{1-6}
 $\att{Flight}$ & $\att{Airline}$ & $\att{Date}$ & $\att{Origin}$ & $\att{Destination}$ & $\att{Airplane}$ \\\cline{1-6}
\val{UA123} & \val{United Airlines} & \val{01/01/2021} & \val{LA} & \val{NY} & \val{N652NW}& $3$ \\
\val{UA123} & \val{United Airlines} & \val{01/01/2021} & \val{NY} & \val{UT} & \val{N652NW}& $2$ \\
\val{UA123} & \val{Delta} & \val{01/01/2021} & \val{LA} & \val{NY} & \val{N652NW}& $1$ \\
\val{DL456} & \val{Southwest} & \val{02/01/2021} & \val{NC} & \val{MA} & \val{N713DX}& $2$ \\
\val{DL456} & \val{Southwest} & \val{03/01/2021} & \val{NJ} & \val{FL} & \val{N245DX}& $1$ \\
\val{DL456} & \val{Delta} & \val{03/01/2021} & \val{CA} & \val{IL} & \val{N819US} & $4$ \\
\cline{1-6}
\end{tabular}
\caption{$D$}
\end{subfigure}
\begin{subfigure}[b]{\linewidth}\centering
\begin{tabular}{|c|c|c|c|c|c|r} 
\cline{1-6}
\rowcolor{Gray}
\multicolumn{6}{l}{\rel{Flights}}\\\cline{1-6}
 $\att{Flight}$ & $\att{Airline}$ & $\att{Date}$ & $\att{Origin}$ & $\att{Destination}$ & $\att{Airplane}$ \\\cline{1-6}
\val{UA123} & \val{United Airlines} & \val{01/01/2021} & \val{NY} & \val{UT} & \val{N652NW}& $2$\\
\val{DL456} & \val{Southwest} & \val{02/01/2021} & \val{NC} & \val{MA} & \val{N713DX}& $2$ \\
\val{DL456} & \val{Southwest} & \val{03/01/2021} & \val{NJ} & \val{FL} & \val{N245DX}& $1$ \\
\cline{1-6}
\end{tabular}
\caption{$E_1$}
\end{subfigure}
\begin{subfigure}[b]{\linewidth}\centering
\begin{tabular}{|c|c|c|c|c|c|r} 
\cline{1-6}
\rowcolor{Gray}
\multicolumn{6}{l}{\rel{Flights}}\\\cline{1-6}
 $\att{Flight}$ & $\att{Airline}$ & $\att{Date}$ & $\att{Origin}$ & $\att{Destination}$ & $\att{Airplane}$ \\\cline{1-6}
\val{UA123} & \val{United Airlines} & \val{01/01/2021} & \val{LA} & \val{NY} & \val{N652NW}& $3$\\
\val{DL456} & \val{Delta} & \val{03/01/2021} & \val{CA} & \val{IL} & \val{N819US} & $4$\\
\cline{1-6}
\end{tabular}
\caption{$E_2$}
\end{subfigure}
\begin{subfigure}[b]{\linewidth}\centering
\begin{tabular}{|c|c|c|c|c|c|r} 
\cline{1-6}
\rowcolor{Gray}
\multicolumn{6}{l}{\rel{Flights}}\\\cline{1-6}
 $\att{Flight}$ & $\att{Airline}$ & $\att{Date}$ & $\att{Origin}$ & $\att{Destination}$ & $\att{Airplane}$ \\\cline{1-6}
\val{UA123} & \val{United Airlines} & \val{01/01/2021} & \val{LA} & \val{NY} & \val{N652NW}& $3$ \\
\val{UA123} & \val{United Airlines} & \val{01/01/2021} & \val{NY} & \val{UT} & \val{N652NW}& $2$\\
\val{DL456} & \val{Delta} & \val{03/01/2021} & \val{CA} & \val{IL} & \val{N819US} & $4$\\
\cline{1-6}
\end{tabular}
\caption{$E_3$}
\end{subfigure}
\caption{\label{fig:DB} For the relation
  $\rel{Flights}(\att{Flight},\att{Airline},\att{Date},\att{Origin},\att{Destination},\att{Airplane})$ and the
  FDs $\att{Flight}\ra\att{Airline}$ (with $w_{\varphi_1}=5$) and
  $\att{Flight Airline Date}\ra\att{Destination}$ (with $w_{\varphi_2}=1$), a database $D$, a cardinality repair $E_1$, a weighted cardinality repair $E_2$, and an optimal subset $E_3$.}
\end{figure}
}

\begin{example}\label{ex:running}
Our running example is based on the database of Figure~\ref{fig:DB} over the relation schema
$\rel{Flights}(\att{Flight},\att{Airline},\att{Date},\att{Origin},\att{Destination},\att{Airplane})$ that contains information about domestic flights in the United States. The weight of each tuple appears on the rightmost column. The FD set $\depset$ consists of the following FDs:
\begin{itemize}
    \item $\att{Flight}\ra\att{Airline}$: a flight is associated with a single airline.
    \item $\att{Flight Airline Date}\ra\att{Destination}$: a flight on a certain date has a single destination.
\end{itemize}
We assume that the weight of the first FD is $5$, and the weight of the second FD is $1$ (as the same flight number can be reused for different flights).

The database $E_1$ of Figure~\ref{fig:DB} is a cardinality repair of $D$ as no repair of $D$ can be obtained by removing less then three facts. However, $E_1$ is not a weighted cardinality repair, since its cost is eight, while the cost of $E_2$ is six. The reader can easily verify that $E_2$ is a weighted cardinality repair of $D$. Finally, $E_3$ is not a repair of $D$ in the traditional sense as it contains a violation of the second FD, but it is an optimal subset of $D$ with $\dbcost(E_3\mid D)=5$.
\qed
\end{example}

 
\eat{
 Note that every cardinality repair is a subset
  repair, but not necesseraly vice versa. A \e{weighted cardinality
  repair} is a consistent subset obtained by solving the following
optimization problem:
$$\min\{\sum_{f\in (D\setminus E)}w_f\mid E\subseteq D,E\models\Delta\}$$
Livshits et al.~\cite{DBLP:journals/tods/LivshitsKR20} have studied the computational complexity of the problem of computing a weighted cardinality repair (which they refer to as an \e{optimal repair}) and obtained a dichotomy classifying FD sets into those for which the problem can be solved in polynomial time, and those for which it is APX-complete (i.e., hard to approximate beyond some constant).

Here, we generalize this problem to account for soft constraints. That is, we assume that violations of the constraints are allowed, but with penalty. The penalty for violating an FD $\varphi$ is the weight $w_\varphi$ associated with the FD. We now formally define the problem. For a subset $E$ of $D$ we denote by $V(\varphi,E)$ the number of pairs of facts violating the FD. That is, 
$$V(\varphi,E)=\frac{1}{2}\left|\set{\set{f,g}\mid f,g\in E, \set{f,g}\not\models\varphi}\right|$$
\nofar{Why do we need this half?}
}
\section{Preliminary Complexity Analysis}\label{sec:complexity}

We consider the data complexity of the problem of computing an optimal subset. We assume that the schema and the set of FDs are fixed, and the input consists of the database.  Livshits et al.~\cite{DBLP:journals/tods/LivshitsKR20} studied the problems of finding a cardinality repair and a weighted cardinality repair, and established a dichotomy over the space of all the sets of functional dependencies.  In particular, they introduced an algorithm that, given a set $\Delta$ of FDs, decides whether:
\begin{enumerate}
\item A weighted cardinality repair can be computed in polynomial time; \e{or}
\item Finding a (weighted) cardinality repair is APX-complete.\footnote{Recall that APX is the class of
    NP optimization problems that admit constant-ratio approximations in polynomial time.
Hardness in APX is via the so called ``PTAS'' reductions (cf.~textbooks on approximation complexity, e.g., \cite{DBLP:reference/crc/2007aam}).}    
\end{enumerate}
No other possibility exists. The algorithm, which is depicted here as Algorithm~\ref{alg:simplify}, is a recursive procedure that attempts to simplify $\Delta$ at each iteration by finding a \e{removable} pair $(X,Y)$ of attribute sets, and removing every attribute of $X$ and $Y$ from all the FDs in $\Delta$ (which we denote by $\depset-XY$).  Note that $X$ and $Y$ may be the same, and then the condition states that every FD contains $X$ on the left hand side. If we are able to transform $\Delta$ to an empty set of FDs by repeatedly applying simplification, then the algorithm returns true and finding an optimal consistent subset is solvable in polynomial time. Otherwise, the algorithm returns false and the problem is APX-complete.  We state their result for later reference.

\begin{algorithm}[t]
\begin{algorithmic}
\STATE Remove trivial FDs from $\depset$
\IF{$\depset$ is not empty}
\STATE find a removable pair $(X,Y)$ of attribute sequences:
\begin{citemize}
\item $\mathrm{Closure_{\Delta}}(X)=\mathrm{Closure_{\Delta}}(Y)$
\item $XY$ is nonempty
 \item  Every FD in $\Delta$ contains either $X$ or $Y$ on the left-hand side
 \end{citemize}
\STATE $\Delta \assn \Delta-XY$
\ENDIF
\end{algorithmic}
\caption{$\algname{Simplify}$\label{alg:simplify}}
\end{algorithm}

\begin{citedtheorem}{DBLP:journals/tods/LivshitsKR20}\label{thm:livshits}
  Let $\Delta$ be a set of FDs. If $\Delta$ can be emptied via $\algname{Simplify()}$ steps, then a weighted cardinality repair can be computed in polynomial time; otherwise, finding a cardinality repair is
  APX-complete.
\end{citedtheorem}

The hardness side of Theorem~\ref{thm:livshits} immediately implies
the hardness of the more general soft-repairing problem.  Yet, the
other direction (tractability generalizes) is not necessarily true. As
discussed in the Introduction, if
$\Delta=\set{ A\rightarrow B, B\rightarrow A, B\rightarrow C}$, then
$\Delta$, as a set of hard constraints, is classified as tractable
according to Algorithm~\ref{alg:simplify}; however, this is not the
case for soft constraints. We can generalize this example by stating
that if $\Delta$ contains a \e{subset} that is hard according to
Theorem~\ref{thm:livshits}, then soft repairing is hard.  (This does not hold when considering only hard constraints, as the example shows
that there exists an easy $\Delta$ with a hard subset.)  In the
following sections, we are going to discuss tractable cases of FD
sets. Before that, we will show that the problem becomes tractable if
one settles for an approximation.


\subsection{Approximation}
The following theorem shows that soft repairing admits a
constant-ratio \e{approximation}, for the constant three, in
polynomial time.  This means that there is a polynomial-time algorithm
for finding a subset with a cost of at most three times the minimum.

\begin{theorem}\label{thm:approx}
For all FD sets, soft repairing admits a 3-approximation in polynomial time.
\end{theorem}
\begin{proof}
  We reduce soft repairing to the problem of finding a minimum
  weighted set cover where every element belongs to $3$ sets. `A simple greedy algorithm finds
  a $3$-approximation to this problem in linear time~\cite{hochbaum1982approximation}.

We set the elements to be $\{(\{f,g\},\delta)\mid f,g\in D, \delta\in\Delta, \text{$f$ and $g$ contradict $\delta$}\}$.
Each element $(\{f,g\},\delta)$ belongs to three sets:
$f$ with weight $w_f$, $g$ with weight $w_g$, and $(\{f,g\},\delta)$ with weight $w_\delta$.
Each minimal solution to this set cover problem can be translated to a soft repair:
the selected sets that correspond to tuples are removed in the repair.
Indeed, a minimal set cover of such a construction has to resolve each conflict by either paying for the removal of at least one of the tuples or paying for the violation.
\end{proof}

In terms of formal complexity, Theorem~\ref{thm:approx} implies that the problem of soft repairing is in APX (for every set of FDs).  From this, from Theorem~\ref{thm:livshits} and from the discussion that follows Theorem~\ref{thm:livshits}, we conclude the following.

\begin{corollary}\label{cor:hardcases-soft-repairing}
  Let $\Delta$ be a set of FDs. Soft repairing for $\Delta$ is in APX. Moreover, if any subset of $\Delta$
  cannot be emptied via $\algname{Simplify()}$ steps, then soft repairing is
  APX-complete for $\Delta$.
  \end{corollary}
  

\section{Algorithm for a Single Functional Dependency}\label{sec:single}
In this section, we consider the case of a single functional
dependency, and present a polynomial-time algorithm for soft
repairing. Hence, we establish the following result.

\begin{theorem}\label{thm:onefd}
  In the case of a single FD, soft repairing can be solved in
  polynomial time.
\end{theorem}
Next, we prove Theorem~\ref{thm:onefd} by presenting an
algorithm. Later, we also generalize the argument and result beyond a
single FD (Theorem~\ref{thm:seq-split}).

We assume that the single FD is $\varphi\eqdef X\rightarrow Y$and that our input database is $D$.  We split $D$ into
\e{blocks} and \e{subblocks}, as we explain next.  The blocks of $D$
are the maximal subsets of $D$ that agree on the $X$ values. Denote
these blocks by $D_1,\dots,D_m$.  Note that there are no conflicts
across blocks; hence, we can solve the problem separately for each
block and then an optimal subset $E$ is simply the union of optimal
subsets $E_i$ of the blocks $D_i$:
$$E=\bigcup_{i=1}^m E_i$$
The subblocks of a block $D_i$ are the maximal subsets of $D_i$
that agree on the $Y$ values (in addition to the $X$ values).  We
denote these subblocks by $D_{i,1},\dots, D_{i,q_i}$. Note that
two facts from the same subblock are consistent, while two facts from
different subblocks are conflicting.

From here we continue with dynamic programming. For a number
$j\in\{0,\dots,q_i\}$, where $q_i$ is the number of subblocks of
$D_i$, and a number $k\in\{0,\dots,|D_{i,1}\cup\dots\cup
D_{i,j}|\}$ of facts, we
define the following values that we are going to compute:
\begin{itemize}
\item $C[i,j,k]$ is the cost of an optimal subset of
  $D_{i,1}\cup\dots\cup D_{i,j}$ (i.e., the union of the first
  $j$ subblocks) with precisely $k$ facts.
\item $F[i,j,k]$ is a subset of  $D_{i,1}\cup\dots\cup
  D_{i,j}$  that realizes $C[i,j,k]$, that is,
 $$|F[i,j,k]|=k\quad\land\quad\dbcost\left(F[i,j,k]\mid D_{i,1}\cup\dots\cup D_{i,j}\right)=C[i,j,k]$$
\end{itemize}
(If multiple choices of $F[i,j,k]$ exist, we select an arbitrary one.)
Once we compute the $F[i,q_i,k]$, we are done since it then suffices
to return the best subset over all $k$:
$$E_i=F[i,q_i,k] \mbox{ for } k=\argmin_k{C[i,q_i,k]}$$

It remains to compute $C[i,j,k]$ and $F[i,j,k]$. We will focus
on the former, as the latter is obtained by straightforward
bookkeeping.  The key observation is that if we decide to delete $t$
facts from $D_{i,j}$, then we always prefer to delete the $t$
facts with the minimal weight. We use this observation as follows.

For a subblock $D_{i,j}$ and $t\in\{0,\dots,|D_{i,j}|\}$, denote by
$\top(t, D_{i,j})$ an arbitrary subset of $D_{i,j}$ with $t$
facts of the highest weight. 
Hence, $\top(t, D_{i,j})$ is obtained
by taking a prefix of size $t$ when sorting the tuples of $D_{i,j}$
from the heaviest to the lightest. Then $C[i,j,k]$ is computed as
follows.

\[
C[i,j,k]=
\begin{cases}
0 & \mbox{$j=0$ and $k=0$};\\
\infty & \mbox{$j=0$ and $k>0$};\\
\underset{t}{\min}\Big(C[i,j-1,k-t]+
 t (k-t) w_\varphi+
\underset{\substack{f\in
    D_{i,j}\setminus \\ \top(t, D_{i,j} )}}{\sum}w_f\Big) & \mbox{otherwise.}
%
\end{cases}
\]
The correctness of the above computation is due to the definition of
the cost in Equation~\eqref{eq:cost}. In particular, in the third case, we go over all options for the number $t$ of facts taken from the subblock $D_{i,j}$ and choose an option with the minimum cost. This cost consists of the following components:
\begin{itemize}
  \item $C[i,j-1,k-t]$ is the cost of the best choice of $k-t$
    facts from the remaining $j-1$ subblocks.
  \item $t (k-t) w_\varphi$ is the cost of the violations in which the $j$th subblock participates: any combination of a fact from $D_{i,j}$ and a fact from the other subblocks is a violation of $\varphi$.
    \item $\sum_{f\in
    D_{i,j}\setminus \top(t, D_{i,j} )}w_f$ is the cost of removing every fact that is not in $\top(t, D_{i,j})$ from the $j$th subblock.
\end{itemize}
This completes the description of the algorithm. From this
description, the correctness should be a straightforward conclusion.

\subsection{Extension}
In this section, we generalize the idea from the previous section.  An
attribute $A$ is an \e{lhs attribute} of an FD $X\rightarrow Y$ if
$A\in X$, and it is a \e{consensus attribute} of $X\rightarrow Y$ if
$X=\emptyset$ and $A\in Y$ (hence, $X\rightarrow Y$ states that all
tuples should have the same $A$ value). The simplification step of
Algorithm~\ref{alg:our-simplify} removes an attribute $A$  if for every FD in $\depset$, it is either an lhs or a consensus attribute.
We prove the following.

\eat{ We say that an attribute $A$ is \e{splittable} w.r.t.~an FD set
  $\depset$ if every FD $X\rightarrow Y$ in $\depset$ meets one of the
  following conditions: \e{(1)} $A\in X$, or \e{(2)} $A\in Y$ and
  $X=\emptyset$. Similarly to Livshits et
  al.~\cite{DBLP:journals/tods/LivshitsKR20}, we introduce a recursive
  procedure (depicted as Algorithm~\ref{alg:our-simplify}) that
  simplifies the FD set $\depset$ at each iteration until $\depset$ is
  either empty or cannot be further simplified. Here, we simplify an
  FD set by removing a splittable attribute $A$ from all the FDs.  }


\begin{algorithm}[t]
\begin{algorithmic}[1]
\STATE remove trivial FDs from $\Delta$
\IF{$\depset$ is not empty}
\STATE find $A$ such that in each FD, $A$ is either  an lhs or  a consensus attribute
\STATE $\Delta\assn\Delta-A$
\ENDIF
\end{algorithmic}
\caption{\algname{L/C-Simplify()}\label{alg:our-simplify}}
\end{algorithm}

\begin{theorem}\label{thm:seq-split}
  Let $\Delta$ be a set of FDs. If $\Delta$ can be emptied via
  \algname{L/C-Simplify()} steps, then soft repairing for $\Delta$ is
  solvable in polynomial time.
\end{theorem}

Note that whenever $\Delta$ can be emptied via
\algname{L/C-Simplify()} steps, it can also be emptied via
\algname{Simplify()} steps. Indeed, if \algname{L/C-Simplify()}
eliminates the attribute $A$, then we can take: \e{(a)} $X=\set{A}$
and $Y=\emptyset$ in Algorithm~\ref{alg:simplify} if $A$ is a
consensus attribute of some FD, \e{or (b)} $X=Y=\set{A}$ if $A$ is an
lhs attribute of every FD.  This is expected due to
Theorems~\ref{thm:livshits} and~\ref{thm:seq-split}, and the
observation of Section~\ref{sec:complexity} that soft-repairing is
hard whenever computing a cardinality repair is hard.

\begin{example}
Consider the database and the FD set of our running example (Example~\ref{ex:running}). This FD set, which we denote here by $\depset_1$, can be emptied via \algname{L/C-Simplify()} steps, by selecting attributes in the following order:
\begin{align*}
   &\set{\att{Flight}\ra\att{Airline}\,,\,\att{Flight Airline Date}\ra\att{Destination}} \\
   \att{Flight}: &\set{\emptyset\ra\att{Airline}\,,\,\att{Airline Date}\ra\att{Destination}}\\
   \att{Airline}: &\set{\att{Date}\ra\att{Destination}}\\
   \att{Date}: &\set{\emptyset\ra\att{Destination}}\\
   \att{Destination}: &\set{}\\
        \end{align*}
Hence, Theorem~\ref{thm:seq-split} implies that soft repairing can be solved in polynomial time for $\depset_1$.

Next, consider the FD set $\depset_2$ consisting of the following FDs: $\att{Flight}\ra\att{Airline}$ and $\att{Flight Date}\ra\att{Destination}$.
This FD set is logically equivalent to $\depset_1$; hence, they both entail the exact same cardinality repairs. However, these sets are no longer equivalent when considering soft repairing. In particular, two facts that agree on the values of the $\att{Flight}$ and $\att{Date}$ attributes, but disagree on the values of the $\att{Airline}$ and $\att{Destination}$ attributes, violate only one FD in $\depset_1$ but two FDs in $\depset_2$, which affects the cost of keeping these two tuples in the database. In fact, the FD set $\depset_2$ cannot be emptied via \algname{L/C-Simplify()} steps, as after removing the $\att{Flight}$ attribute, no other attribute is either an lhs or a consensus attribute of the remaining FDs. The complexity of soft repairing for $\depset_2$ remains an open problem.\qed
\end{example}

Next, we prove Theorem~\ref{thm:seq-split} by presenting a
polynomial-time algorithm for soft repairing in the case where
$\Delta$ can be emptied via \algname{L/C-Simplify()} steps.  Our
algorithm generalizes the idea of the algorithm for a single FD, and
we again use dynamic programming.

The main observation is as follows. Let $A$ be an attribute chosen by \algname{L/C-Simplify()}, and let
$D_1,\dots,D_m$ be the maximal subsets of $D$ that agree on the value
of $A$, which we refer to as blocks (w.r.t.~$A$). Two facts from
different blocks violate \e{all} of the FDs wherein $A$ is a consensus
attribute and \e{none} of the FDs wherein  $A$ is an
lhs attribute.
Therefore, to compute the cost of a soft repair, each pair of facts from different blocks is charged with the violation of all FDs wherein $A$ is a consensus attribute. Then, we can remove $A$ from all FDs and continue the computation separately for each block.

\eat{
\begin{itemize}
\item If $A$ is an lhs attribute of every FD (and a consensus
  attribute of none of the FDs), then facts from different subsets
  $D_i$ \e{never} jointly violate an FD of $\depset$ and we can solve the
  problem separately for each one of the subsets, as we have done in
  the case of a single FD.
\item Otherwise, $A$ is a consensus attribute of at least one FD, and two
  facts from different subsets always jointly violate every FD of the
  form $\emptyset\rightarrow Y$ such that $A\in Y$, but none of the
  other FDs of $\depset$; hence, the cost of a violation among facts
  from different subsets is constant.
\end{itemize}
}
  
Now, let $\depset$ be an FD set that can be emptied via
\algname{L/C-Simplify()} steps, and let $A_1,\dots,A_n$ be the
attributes in the order of such an elimination process. For each
$\ell\in\set{1,\dots,n+1}$, we denote by $\depset_\ell$ the FD set in
line~2 of the $\ell$th iteration of this execution (after removing the trivial FDs). Thus, $\depset_1$
contains every non-trivial FD of $\depset$,
and $\depset_{n+1}$ is empty.
We also denote by
$w_\ell$ the total weight of the FDs in
$\depset_{\ell}$ of which $A$ is a consensus attribute (if there are no such FDs, then $w_\ell=0$).  

In the algorithm for a single FD, the recursion steps were with
respect to the block $D_i$ (which determines the value of $X$), and so
the value of $i$ was a parameter. Here, we need to maintain the
assignment $\tau$ to all previously handled attributes, and we use
$\tau$
and $\ell$ as parameters.
Given $1\leq \ell \leq n+1$,
if $\tau$ is an assignment to the
attributes $A_1,\dots,A_{\ell-1}$, then $D^{\tau}$
denotes the database
$\sigma_{\tau}D$ (i.e., the database that contains all the tuples that
agree with $\tau$ on the values of the attributes
$A_1,\dots,A_{\ell-1}$).
We denote by $D^{\tau}_{1},\dots,D^{\tau}_{q_\ell^\tau}$ the blocks of $D^{\tau}$ w.r.t.~$A_\ell$. Moreover, we denote by $\tau\wedge (A_{\ell}=j)$ the assignment to the
attributes $A_1,\dots,A_{\ell}$ that agrees with block $D^{\tau}_{j}$ on the value assigned to $A_\ell$ and agrees with $\tau$ on all other values.
We denote by $F[\ell,\tau,j,k]$ an optimal
subset of $D^\tau_{1}\cup\dots\cup D^\tau_{j}$ of size $k$
w.r.t.~$\depset_{\ell}$.
We also denote by $C[\ell,\tau,j,k]$ the cost of
$F[\ell,\tau,j,k]$. 
According to Equation~\eqref{eq:cost}, our goal is
to compute $F[1,\emptyset,q_1^\emptyset,k]$ for
$k=\argmin_k{C[1,\emptyset,q_1^\emptyset,k]}$.

We again focus on the computation of $C[\ell,\tau,j,k]$ that can be
done as follows.

\[
C[\ell,\tau,j,k]=
\begin{cases}
\underset{\substack{f\in
    D^\tau\setminus\top(k, D^\tau )}}{\sum}w_f & \ell=n+1,\\
0 & j=0, k=0,\\
\infty & j=0, k>0,\\
\parbox[t]{.5\textwidth}{
$\underset{t}{\min}\Big(C[\ell,\tau,j-1,k-t]+
 t (k-t) w_\ell+\\
\mbox{~~~~~~~} C[\ell+1,\tau\wedge (A_\ell=j),q_{\ell+1}^{\tau\wedge( A_\ell=j)},t]
\Big)$} & \mbox{otherwise}.
\end{cases}
\]

The first line (where $\ell=n+1$) refers to the case where $\depset$ is empty. Since there are no FDs that need to be taken into account, the optimal subset of $D^\tau$ of size $k$ consists of the $k$ facts of the highest weight.
In the fourth case, we go over all options for the number $t$ of facts taken from the block $D^\tau_{j}$ and choose an option with the minimum cost. This cost consists of the following components:
\begin{itemize}
  \item $C[\ell,\tau,j-1,k-t]$ is the cost of the best choice of $k-t$ facts from the remaining $j-1$ blocks.
  \item $t (k-t) w_\ell$ is the cost of the violations in which the $j$th block participates: any combination of a fact from $D^\tau_{j}$ and a fact from the other blocks $D^\tau_{1}\cup\dots\cup D^\tau_{j-1}$ is a violation of the FDs in which $A$ is a consensus attribute.
  \item $C[\ell+1,\tau\wedge (A_\ell=j),q_{\ell+1}^{\tau\wedge( A_\ell=j)},t]$ is the cost of the further repairing needed following the elimination of $A_\ell$ (i.e., repairing with respect to $\depset_{\ell+1}$) applied to the current block (the $t$ facts from $D^\tau_j$) .
  \end{itemize}

The given recursion can be computed in polynomial time via dynamic programming; thus, this proves Theorem~\ref{thm:seq-split}.

\section{Algorithm for Matching Constraints}\label{sec:matching}

Next, we consider the case of a ``matching'' constraint, where the FD set $\Delta$ states two keys that cover all of the attributes. (We give the precise definition in Section~\ref{sec:matching-generalization}.)  We present a polynomial-time algorithm for soft repairing in this case. For presentation sake, we first describe the algorithm for the special case where the schema is $R(A,B)$ and $\Delta\eqdef\set{A\rightarrow B,B\rightarrow A}$.  Later in the section, we generalize it to the case of two keys.  So, we begin by proving the following lemma.

\begin{lemma}\label{lemma:ABA} Soft repairing is solvable in polynomial time for $R(A,B)$ and $\Delta=\set{A\rightarrow B,B\rightarrow A}$.  \end{lemma}

\begin{figure}[t]
\centering
  \begin{subfigure}[]{0.3\linewidth}
  \centering
        \begin{tabular}{c||c|c}
        & $A$ & $B$\\\hline
        $f_1$ & $a_1$ & $b_1$ \\
        $f_2$ & $a_1$ & $b_2$ \\
        $f_3$ & $a_1$ & $b_3$ \\
        $f_4$ & $a_2$ & $b_1$ \\
        $f_5$ & $a_2$ & $b_2$\\
        $f_6$ & $a_3$ & $b_3$
    \end{tabular}
    \caption{Database\label{fig:database}}
  \end{subfigure}
\begin{subfigure}[]{0.5\linewidth}
\centering
  \input{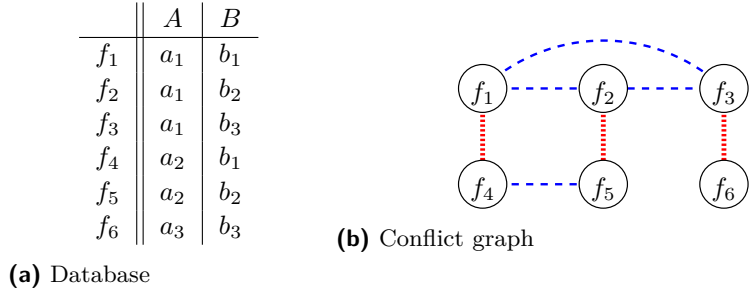}
  \caption{Conflict graph\label{fig:conflict}}
  \end{subfigure}
  \caption{A database over $R(A,B)$ and its conflict graph w.r.t.~$\set{A\rightarrow B, B\rightarrow A}$.}
\end{figure}

In the remainder of this section, we assume the input $D$ over $R(A,B)$.
We begin with an observation.
For $E\subseteq D$ it holds that:
$$\sum_{f\in(D\setminus E)}w_f=\sum_{f\in D}w_f-\sum_{f\in E}w_f$$
Since the value $\sum_{f\in D}w_f$ does not depend on the choice of $E$, minimizing the value $\left(\sum_{f\in (D\setminus E)}\!\!w_f\right)
  + \left(\sum_{\varphi\in\Delta}w_\varphi|\vio(E,\varphi)|\right)$ is the same as minimizing the value $\left(\sum_{f\in E}\!\!-w_f\right)
  + \left(\sum_{\varphi\in\Delta}w_\varphi|\vio(E,\varphi)|\right)$. We use the following notation:
$$w_D(E)=\left(\sum_{f\in E}\!\!-w_f\right)
  + \left(\sum_{\varphi\in\Delta}w_\varphi|\vio(E,\varphi)|\right)$$

To solve the problem, we construct a reduction to the \e{Minimum Cost
  Maximum Flow} (MCMF) problem.
The input to MCMF is a flow network $\N$, that is, a directed graph $(V,E)$ with a \e{source} node $s$ having no incoming edges and a \e{sink} node $t$ having no outgoing edges. Each edge $e\in E$ is associated with a \e{capacity} $c_e$ and a \e{cost} $c(e)$. A flow
$f$ of $\N$ is a function $f:E\rightarrow\mathbb{R}$ such that
$0\le f(e)\le
c_e$ for every $e\in E$, and moreover, for every node $v\in V\setminus\set{s,t}$
it holds that $\sum_{e\in I_v} f(e)=\sum_{e\in O_v} f(e)$ where $I_v$ and $O_v$ are the sets of incoming and outgoing edges of $v$, respectively.  A \e{maximum flow} is a flow $f$ that maximizes the value $\sum_{(s,v)\in E} f(s,v)$, and a \e{minimum cost maximum flow} is a maximum flow $f$ with a minimal cost, where the cost of a flow is defined by $\sum_{e\in E}f(e)\cdot c(e)$.  We say that $f$ is \e{integral} if all values $f(e)$ are integers.  It is known that, whenever the capacities are integral (i.e., natural numbers, as will be in our case), an integral minimum cost maximum flow exists and, moreover, can be found in polynomial time~\cite[Chapter 9]{DBLP:books/daglib/0069809}.


From $D$ we construct $n$ instances $\N_1,\dots,\N_n$ of the MCMF problem, where $n$ is the number of facts in $D$, in the following way.

First, we denote the FD $A\rightarrow B$ by $\varphi_1$ and the FD $B\rightarrow A$ by $\varphi_2$. We also denote by $D.A$ the set of values occurring in attribute $A$ in $D$ (that is, $D.A=\set{\val{a}\mid \exists f\in D (f[A]=\val{a})}$). We do the same for attribute $B$ and denote by $D.B$ the set of values that occur in attribute $B$ in $D$. For each value $\val{a}\in D.A$ we denote by $\#_{D.A}(\val{a})$ the number of appearances of the value $\val{a}$ in attribute $A$ (i.e., the number of facts $f\in D$ such that $f[A]=\val{a}$). Similarly, we denote by $\#_{D.B}(\val{b})$ the number of appearances of the value $\val{b}$ in attribute $B$ in $D$. Observe that 
$$\vio(D,\varphi_1)=\frac{1}{2}\cdot\sum_{\val{a}\in D.A}\left[\#_{D.A}(\val{a})\cdot(\#_{D.A}(\val{a}) - 1)\right]$$ 
since every fact of the form $R(\val{a},\val{b})$ violates $\varphi_1$ with every fact $R(\val{a},\val{c})$ where $\val{b}\neq \val{c}$. Similarly, it holds that $$\vio(D,\varphi_2)=\frac{1}{2}\cdot\sum_{\val{b}\in D.B}\left[\#_{D.B}(\val{b})\cdot(\#_{D.B}(\val{b}) - 1)\right]$$

    \begin{figure}[t]
      \centering
  \input{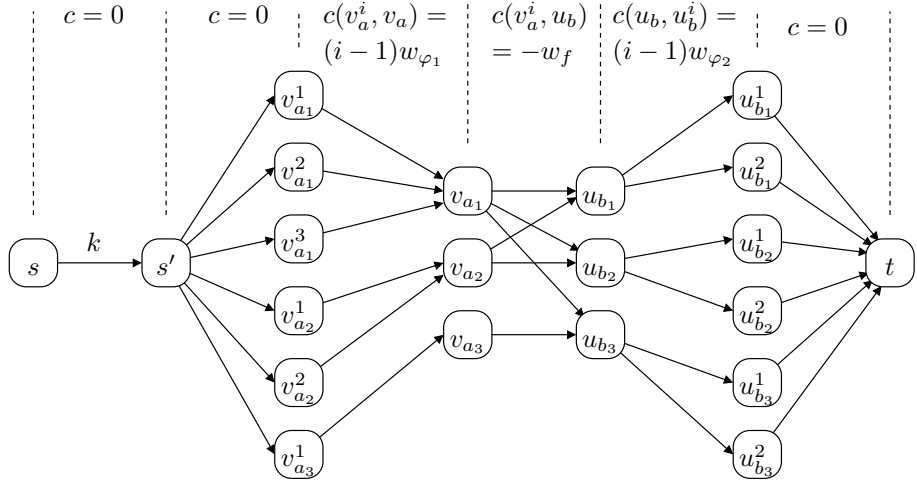}
  \caption{The network $\N_k$ constructed from the database of Figure~\ref{fig:database}. The capacity of all edges is $1$, except for the edge $(s,s')$ that has capacity $k$.\label{fig:network}}
\end{figure}

Next, we describe the construction of the network $\N_k$. Our
construction for the database of Figure~\ref{fig:database} is
illustrated in Figure~\ref{fig:network}. Note that
Figure~\ref{fig:conflict}
depicts the conflict graph of the database of Figure~\ref{fig:database} w.r.t.~$\Delta=\set{A\rightarrow B,B\rightarrow A}$, which contains a vertex for each fact in the database and an edge between two vertices if the corresponding facts jointly violate an FD of $\Delta$. The blue edges in the conflict graph are violations of the FD $A\rightarrow B$ and the red edges are violations of the FD $B\rightarrow A$.

For each $k\in\set{1,\dots,n}$ we construct the network $\N_k$ that consists of the set $\set{s,s',t}\cup V\cup A\cup B\cup U$ of nodes where:
\begin{itemize}
    \item $A=\set{v_\val{a}\mid \val{a}\in D.A}$
    \item $B=\set{u_\val{b}\mid \val{b}\in D.B}$
    \item $V=\set{v_\val{a}^i\mid\val{a}\in D.A, 1\le i\le \#_{D.A}(\val{a})}$
    \item $U=\set{u_\val{b}^i\mid\val{b}\in D.B, 1\le i\le \#_{D.B}(\val{b})}$
\end{itemize}

$\N_k$ contains the following edges:
\begin{itemize}
    \item $(s,s')$, with cost $c(s,s')=0$
    \item $(s',v_\val{a}^i)$ for every $v_\val{a}^i\in V$, with cost $c(s',v_\val{a}^i)=0$
    \item $(v_\val{a}^i,v_\val{a})$ for every value $\val{a}\in D$, with cost $c(v_\val{a}^i,v_\val{a})=(i-1)\cdot w_{\varphi_1}$
    \item $(v_\val{a},u_\val{b})$ for every $\val{a}\in D.A$ and
      $\val{b}\in D.B$ such that $f=R(\val{a},\val{b})$ occurs in $D$,
      with cost $c(v_\val{a},u_\val{b})=-w_f$
    \item $(u_\val{b},u_\val{b}^i)$ for every value $\val{b}\in D$, with cost $c(u_\val{b},u_\val{b}^i)=(i-1)\cdot w_{\varphi_2}$
    \item $(u_\val{b}^i,t)$ for every $u_\val{b}^i\in U$, with cost $c(u_\val{b}^i,t)=0$
    \end{itemize}

The capacity of the edge $(s,s')$ is $k$ and the capacity of the other edges is $1$.
The intuition for the construction is as follows. A network with edges of the form $(v_\val{a},u_\val{b})$ that are connected to a source on one side and a target on the other corresponds to a matching, which in turn corresponds to a traditional repair. To allow violations of $A\rightarrow B$, we add the vertices $v_\val{a}^i$. The cost of a violation of this FD is defined by the cost of the edges $(v_\val{a}^i,v_\val{a})$. In particular, if we keep $k$ facts of the form $R(\val{a},\cdot)$ for some $\val{a}\in D.A$ we pay $\sum_{i=1}^k(k-1)w_{\varphi_1}$ for violations of $\varphi_1$. We include the vertices $v_\val{b}^i$ to similarly allow violations of $B \rightarrow A$. The discarding of facts is discouraged by offering gain for the edges $(v_\val{a},u_\val{b})$. Finally, to prevent the case where the flow always fills the entire network (which corresponds to taking all facts and paying for all violations), we introduce the edge $(s,s')$ which limits the capacity of the network, and enables us to find the minimum cost flow of a given size $k$.
We will show that for every $k$, the cost of the solution to the MCMF problem on $\N_k$ will be the cost of the ``cheapest'' subinstance of $D$ of size $k$. Hence, the solution to our problem is the cost of the minimal solution among all the instances $\N_1,\dots,\N_n$.

Given an integral flow $f$ in $\N_k$, the repair $D[f]$ induced by $f$, is the set of facts $R(\val{a},\val{b})$ corresponding to edges of the form $(v_\val{a},u_\val{b})$ such that $f(v_\val{a},u_\val{b})=1$. Moreover, given a subinstance $E$ of $D$ of size $k$, we denote by $f_E$ the integral flow in $\N_k$ defined as follows.
\begin{itemize}
    \item $f_E(s,s')=k$
    \item $f_E(s',v_\val{a}^i)=1$ for $1\le i\le \#E.A(\val{a})$ and $f_E(s',v_\val{a}^i)=0$ for $i>\#E.A(\val{a})$ for every $\val{a}\in E.A$
    \item $f_E(v_\val{a}^i,v_\val{a})=1$ for $1\le i\le \#E.A(\val{a})$ and $f_E(v_\val{a}^i,v_\val{a})=0$ for $i>\#E.A(\val{a})$ for every $\val{a}\in E.A$
    \item $f_E(v_\val{a},u_\val{b})=1$ if $R(\val{a},\val{b})\in E$ and $f_E(v_\val{a},u_\val{b})=0$ otherwise
    \item $f_E(u_\val{b},u_\val{b}^i)=1$ for $1\le i\le \#E.B(\val{b})$ and $f_E(u_\val{b},u_\val{b}^i)=0$ for $i>\#E.B(\val{b})$ for every $\val{b}\in E.B$
    \item $f_E(u_\val{b}^i,t)=1$ for $1\le i\le \#E.B(\val{b})$ and $f_E(u_\val{b}^i,t)=0$ for $i>\#E.B(\val{b})$ for every $\val{b}\in E.B$
\end{itemize}
The reader can easily verify that $f_E$ is indeed an integral flow in $\N_k$. Clearly, the value of the flow is $k$.

We have the following lemmas. The first is proved in the Appendix and the second follows straightforwardly from the construction of $\N_k$ and the definition of $f_E$.


\def\lemmaMCMFsecond{
Every integral solution $f$ to MCMF on $\N_k$ satisfies $\cost(f)=w_D(f[D])$.
}
\begin{lemma}\label{lemma:fcostw}
\lemmaMCMFsecond
\end{lemma}


\begin{lemma}\label{lemma:ecostf}
Every subinstance $E$ of $D$ satisfies $\cost(f_E)=w_D(E)$.
\end{lemma}


Now, let $E$ be an optimal subset of $D$ w.r.t.~$\Delta$ and assume that $|E|=k$. 
Let $f^*$ be a solution with the minimum cost among all the solutions to MCMF on $\N_1\ldots,\N_n$.
Lemma~\ref{lemma:ecostf} implies that there is an integral flow $f_E$ in $\N_k$ such that $\cost(f_E)=w_D(E)$. Hence, we have that $\cost(f^*)\le w_D(E)$.
By applying Lemma~\ref{lemma:fcostw} on $f^*$, there is another subinstance $E'$ of $D$ such that $w_D(E')=\cost(f^*)$.
Since $E$ is an optimal subset, we have that $w_D(E) \le w_D(E')$.
Overall, we have that $\cost(f^*)\le w_D(E)\le w_D(E')=\cost(f^*)$, and we conclude that $\cost(f^*) = w_D(E)$.
Therefore, by taking the solution with the lowest cost among all solutions to MCMF on $\N_1,\dots,\N_n$, we indeed find a solution to our problem, and that concludes our proof of Lemma~\ref{lemma:ABA}.

\begin{figure}
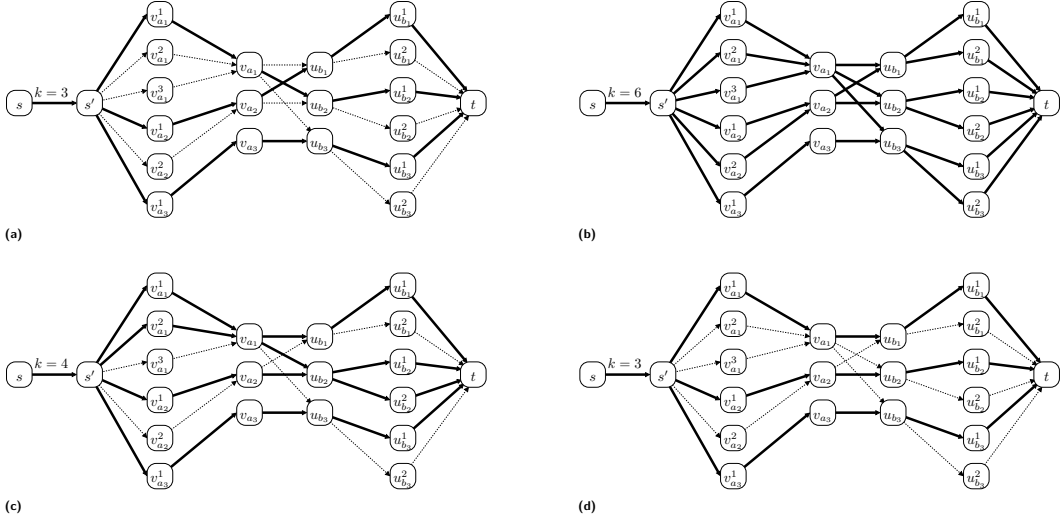

  \centering
    \scalebox{0.53}{
    \begin{subfigure}[]{\linewidth}
        \input{flowA.pspdftex}
        \subcaption{\label{fig:flow1}}
    \end{subfigure}
    }
    \scalebox{0.53}{
    \begin{subfigure}[]{\linewidth}
    \input{flowB.pspdftex}
    \subcaption{\label{fig:flow2}}
    \end{subfigure}
    }    \vskip1em
    \scalebox{0.53}{
    \begin{subfigure}[]{\linewidth}
    \input{flowC.pspdftex}
    \subcaption{\label{fig:flow3}}
    \end{subfigure}
    }
    \scalebox{0.53}{
    \begin{subfigure}[]{\linewidth}
    \input{flowD.pspdftex}
    \subcaption{\label{fig:flow4}}
    \end{subfigure}
    } 
    \caption{The flow in the network $\N_k$ corresponding to an optimal subset of the database of Figure~\ref{fig:database} for different weights.\label{fig:flows}}
\end{figure}

\begin{example}
Consider again the database of Figure~\ref{fig:database}. Assume that:
$$w_{\varphi_1}=w_{\varphi_2}=2 \quad\quad\quad w_{f_1}=w_{f_2}=w_{f_3}=w_{f_4}=w_{f_5}=w_{f_6}=1$$
Since the cost of a violation is ``too high'' in this case (i.e., it is always cheaper to delete a fact involved in a violation than to keep the violation), an optimal subset in this case is, in fact, an optimal repair in the traditional sense (that is, when the constraints are assumed to be hard constraints). One possible optimal repair in this case is $\set{f_2,f_4,f_6}$. The flow corresponding to this repair in the network $\N_3$ is illustrated in Figure~\ref{fig:flow1}.

Now, assume that:
$$w_{\varphi_1}=w_{\varphi_2}=1  \quad\quad\quad
w_{f_1}=w_{f_2}=w_{f_3}=w_{f_4}=w_{f_5}=w_{f_6}=3$$
In this case, the cost of deleting a fact is ``too high'', since each fact is involved in at most two violations, and the cost of keeping the violation is lower than the cost of removing facts involved in the violation. Therefore, the database itself is an optimal subset, and the corresponding flow in the network $\N_6$ is illustrated in Figure~\ref{fig:flow2}.

As another example, assume that:
$$w_{\varphi_1}=w_{\varphi_2}=1  \quad\quad\quad w_{f_1}=w_{f_2}=w_{f_5}=2,w_{f_3}=w_{f_4}=1,w_{f_6}=3$$
Here an optimal subset consists of the facts in $\{f_1,f_2,f_5,f_6\}$, and the corresponding flow in the network $\N_4$ is illustrated in Figure~\ref{fig:flow3}.
If we modify the weight of $\varphi_2$ and define $w_{\varphi_2}=4$, while keeping the rest of the weight intact, it is now cheaper to delete the fact $f_2$ rather than keep the violations it is involved in with $f_1$ and $f_5$; hence, an optimal subset in this case is $\{f_1,f_5,f_6\}$, and the corresponding flow in the network $\N_3$ is illustrated in Figure~\ref{fig:flow4}.\qed
\end{example}

Note that the FD set $\set{A\rightarrow B}$ over $R(A,B)$ is in fact a special case of the result of Theorem~\ref{thm:matching}, as we can compute an optimal subset for this FD set using the algorithm described above by defining $w_{B\rightarrow A}=0$. However, this algorithm works only for the case where the single FD is a key and fails to compute the correct solution when the schema contains attributes that do not appear in the FD. The algorithm described in the proof of Theorem~\ref{thm:onefd}, on the other hand, can handle this case and does not assume anything about the underlying schema.

\subsection{Generalization into Matching Constraints}\label{sec:matching-generalization}
By a ``matching constraint'' we refer to the case of $\hat\Delta=\set{X\rightarrow Y,X'\rightarrow Y'}$ over a schema $\hat R(A_1,\dots,A_k)$ where $X\cup Y=X'\cup Y'=X\cup X'=\set{A_1,\dots,A_k}$. An example follows.
\begin{example}
Consider the database of our running example (Figure~\ref{fig:DB}), and the following FDs:
\begin{itemize}
    \item $\att{Flight Airline Date}\ra\att{Origin Destination Airplane}$,
    \item $\att{Origin Destination Airplane Date}\ra\att{Flight Airline}$.
\end{itemize}
The reader can easily verify that these two FDs form a matching constraint.
On the other hand, consider the set consisting of the following two FDs:
\begin{itemize}
    \item $\att{Flight Date}\ra\att{Airline Origin Destination Airplane}$,
    \item $\att{Origin Destination Airplane Date}\ra\att{Flight Airline}$.
\end{itemize}
Here, we do not have a matching constraint since while it holds that $X\cup Y=X'\cup Y'=\set{\att{Flight},\att{Airline},\att{Date},\att{Origin},\att{Destination},\att{Airplane}}$, the set $X\cup X'$ misses the $\att{Airline}$ attribute.\qed
  \end{example}

  The generalization of Lemma~\ref{lemma:ABA} from $\Delta=\set{A\rightarrow B,B\rightarrow A}$ over $R(A,B)$ to the general case of a matching constraint is fairly straightforward. Given an input $\hat D$ for soft repairing over $\hat R$ and $\hat\Delta$, we construct an input $D$ over $R$ and $\Delta$ by defining unique values $a(\pi_X(\hat f))$ and $b(\pi_{X'}(\hat f))$ for the projections $\pi_X(\hat f)$ and $\pi_{X'}(\hat f)$ over $X$ and $X'$, respectively, of every fact $\hat f$ of $\hat D$. Then, the database $D$ is simply the set of all the pairs $a(\pi_X\hat f)$ and $b(\pi_{X'}\hat f)$ for all facts $\hat f$ of $D$: \[D\eqdef \set{(a(\pi_X\hat f), b(\pi_{X'}\hat f))\mid \hat f\in\hat D}\]
  In addition, we define
  $w_f \eqdef w_{\hat f}$ whenever $f=(a(\pi_X\hat f), b(\pi_{X'}\hat f))$ and 
  $w_{A\rightarrow B}\eqdef w_{X\rightarrow Y}$ and $w_{B\rightarrow A}\eqdef w_{X'\rightarrow Y'}$.
  Note that the mapping $f\rightarrow \hat f$ is reversible since $X\cup X'=\set{A_1,\dots,A_k}$. So, in order to solve
  soft repairing for $\hat D$, we solve it for $D$ and transform every fact $f$ of $D$ into the corresponding fact $\hat f$ of $\hat D$.
  We get the following result. The proof (given in the Appendix) is by showing the correctness of the reduction.
  
  \def\thmmatching{
  Soft repairing is solvable in polynomial time whenever $\Delta$ is a pair of FDs that constitutes a matching constraint.  
  }
  \begin{theorem}\label{thm:matching}
    \thmmatching
  \end{theorem}


\section{Conclusions and Open Problems}\label{sec:conclusions}
We studied the complexity of soft repairing for functional dependencies, where the goal is to find
 an optimal subset under penalties of deletion and constraint
 violation. The problem is harder than that computing a cardinality repair, and we have developed two new, nontrivial algorithms solving natural special cases.
 A full classification of the FD sets remains an open challenge for future research; specifically, the question is what fragment of the positive side of the dichotomy of
 Livshits et al.~\cite{DBLP:journals/tods/LivshitsKR20} remains positive when softness is allowed.
 We have also shown that the problem becomes tractable if we settle for a 3-approximation.

\paragraph*{Open Problems}
 Several directions are left open for future work. A direct open problem is to characterize the class of tractable FDs via a full dichotomy. The simplest sets of FDs where the complexity of soft repairing is open are the following:
 \begin{itemize}
 \item $\set{A\rightarrow B, A\rightarrow C}$. Note that this problem is different from $\set{A\rightarrow BC}$ that consists of a single FD.
 \item $\set{A\rightarrow B, B\rightarrow A}$ in the case where the schema has attributes different from $A$ and $B$, starting with $R(A,B,C)$.
 \item $\set{\emptyset\rightarrow A,B\rightarrow C}$.
 \end{itemize}

The problem is also open for classes of constraints that are more general than FDs, including equality-generating dependencies (EGDs), denial constraints, and inclusion dependencies. Yet, the problem for these types of dependencies is open already in the case of cardinality repairs, with the exception of some cases of EGDs~\cite{DBLP:journals/corr/abs-1904-06492}. Another clear direction is that of \e{update repairs} where we are allowed to change cell values instead of (or in addition to) deleting tuples and where complexity results are known for hard constraints~\cite{DBLP:conf/icdt/KolahiL09, DBLP:journals/tods/LivshitsKR20}.

\eat{
\section{Open Problems}
$\set{A\rightarrow B, A\rightarrow C}$:
This problem is open. Note that this problem is different from that of $\set{A\rightarrow BC}$ that is covered in the case of a single FD.

$\set{A\rightarrow B, B\rightarrow A}$ where there is also the attribute $C$.

$\set{\emptyset\rightarrow A,B\rightarrow C}$
}


\bibliography{main}

\newpage
\appendix\label{sec:appendix}
\section{Details for Section~\ref{sec:matching}}
In this section, we provide the missing proofs of Section~\ref{sec:matching}. For convenience, we give the results again here.


\begin{replemma}{\ref{lemma:fcostw}}
\lemmaMCMFsecond
\end{replemma}
\begin{proof}
First, note that it cannot be the case that $f(s',v_\val{a}^j)=0$ while $f(s',v_\val{a}^i)=1$ for some $j<i$ and $i\in\set{1,\dots,\#_{D.A}(\val{a})}$. Otherwise, we can construct a different integral flow $f'$ with $f'(s',v_\val{a}^j)=f'(v_\val{a}^j,v_\val{a})=1$, $f'(s',v_\val{a}^i)=f'(v_\val{a}^i,v_\val{a})=0$, and $f'(e)=f(e)$ for every other edge $e$.
It holds that $\cost(f')=\cost(f)-c(v_\val{a}^i,v_\val{a})+c(v_\val{a}^j,v_\val{a})$, and since $c(v_\val{a}^i,v_\val{a})>c(v_\val{a}^j,v_\val{a})$ we will have that $\cost(f')<\cost(f)$ in contradiction to the fact that $f$ is a solution to MCMF on $\N_k$. Therefore, for every $\val{a}\in D.A$, if the flow entering the node $v_\val{a}$ is $\ell$, then $f(s',v_\val{a}^i)=f(v_\val{a}^i,v_\val{a})=1$ if $i\le\ell$ and $f(s',v_\val{a}^i)=f(v_\val{a}^i,v_\val{a})=0$ otherwise.
Thus, the total cost of the edges of the form $(v_\val{a}^i,v_\val{a})$ is $\sum_{i=1}^\ell \left[(i-1) w_{\varphi_1}\right] = \frac{1}{2}\ell(\ell-1)w_{\varphi_1}$.
By the definition of $f[D]$, there are $\#_{f[D].A}(\val{a})$ edges of the form $(v_\val{a},u_\val{b})$ for which $f(v_\val{a},u_\val{b})=1$.
By the definition of a flow, this is also the flow entering the node $v_\val{a}$, and we have that $\ell=\#_{f[D].A}(\val{a})$.
We conclude that the total cost of the flow on edges of the form $(v_\val{a}^i,v_\val{a})$ is 
$\sum_{\val{a}\in f[D].A}\left[\frac{1}{2}\cdot\#_{f[D].A}(\val{a})\cdot(\#_{f[D].A}(\val{a}) - 1)\cdot w_{\varphi_1}\right] = \vio(f[D],\varphi_1)\cdot w_{\varphi_1}$.
The same argument shows that the total cost of the flow on edges of the form $(u_\val{b},u_\val{b}^i)$ is $\vio(f[D],\varphi_2)\cdot w_{\varphi_2}$.

\eat{
Similarly, the total cost of the edges of the form $(u_\val{b},u_\val{b}^i)$ for which $f(u_\val{b},u_\val{b}^i)=1$ is 
$\vio(f[D],\varphi_2)\cdot w_{\varphi_2}$.
Indeed, for every $\val{b}\in f[D].B$, the flow entering the node $u_\val{b}$ is $\#_{f[D].B}(\val{b})$. We can again show that if $f(u_\val{b},u_\val{b}^i)=1$ for some $i$, then it is also the case that $f(u_\val{b},u_\val{b}^j)=1$ for every $j<i$. Hence, for every value $\val{b}$, the total cost of the edges of the form $(u_\val{b},u_\val{b}^i)$ for which $f(u_\val{b},u_\val{b}^i)=1$ is $\sum_{i=1}^{\#_{f[D].B}(\val{b})} \left[(i-1)\cdot w_{\varphi_2}\right]=\frac{1}{2}\#_{f[D].B}(\val{b})\cdot(\#_{f[D].B}(\val{b}) - 1)\cdot w_{\varphi_2}$.
Therefore, the total cost of the edges of the form $(u_\val{b},u_\val{b}^i)$ for which $f(u_\val{b},u_\val{b}^i)=1$ is 
$\sum_{\val{b}\in f[D].B}\left[\#_{f[D].B}(\val{b})\cdot(\#_{f[D].B}(\val{b}) - 1)\cdot w_{\varphi_2}\right] = \vio(f[D],\varphi_2)\cdot w_{\varphi_2}$.
}

Finally, the total cost of the edges of the form $(v_\val{a},u_\val{b})$ is $\sum_{g\in f[D]}{(-w_g)}$ by the definition of $f[D]$ and the construction of the network.
We conclude that:
$$\cost(f)=\left(\sum_{g\in f[D]}{(-w_g)}\right)+\vio(f[D],\varphi_1)\cdot w_{\varphi_1}+\vio(f[D],\varphi_2)\cdot w_{\varphi_2}$$
and $\cost(f)=w_D(f[D])$ by definition.
\end{proof}

\begin{reptheorem}{\ref{thm:matching}}
\thmmatching
\end{reptheorem}
  \begin{proof}
  We prove that $D$ has a subset $E$ with $\cost(E\mid D)=k$ if and only if $\hat D$ has a subset $\hat E$ with $\cost(\hat E\mid \hat D)=k$.
  Let $E$ be a subset of $D$ with cost $k$. Let $\hat E$ be a subset of $\hat D$ that includes the fact $\hat f$ for every $f\in E$. By definition, we have that $\sum_{f\in(D\setminus E)w_f}=\sum_{f\in(\hat D\setminus \hat E)w_{\hat f}}$; hence, it is left to show that $\sum_{\varphi\in\depset}w_\varphi|\vio(E,\varphi)|=\sum_{\hat\varphi\in\hat\depset}w_{\hat\varphi}|\vio(\hat E,\hat\varphi)|$. Let $f,g\in E$ such that $\set{f,g}\not\models (A\rightarrow B)$. Hence, it holds that $f[A]=g[A]$ while $f[B]\neq g[B]$. From the construction of $D$, we have that $\pi_X\hat f=\pi_X\hat g$, while $\pi_{X'}\hat f\neq\pi_{X'}\hat g$. Thus, there is an attribute $A_i\in X'$ such that $\hat f[A_i]\neq \hat g[A_i]$ and since $A_i\not\in X$ and $X\cup Y=\set{A_1,\dots,A_k}$, it holds that $A_i\in Y$. We conclude that $\set{\hat f,\hat g}\not\models (X\rightarrow Y)$. We can similarly prove that if $\set{f,g}\not\models (B\rightarrow A)$, then $\set{\hat f,\hat g}\not\models (X'\rightarrow Y')$. Finally, because $w_{A\rightarrow B}=w_{X\rightarrow Y}$ and $w_{B\rightarrow A}=w_{X'\rightarrow Y'}$ it holds that $\sum_{\varphi\in\depset}w_\varphi|\vio(E,\varphi)|=\sum_{\hat\varphi\in\hat\depset}w_{\hat\varphi}|\vio(\hat E,\hat\varphi)|$.
  
  For the other direction, let $\hat E$ be a subset of $\hat D$, and let $E$ be the subset of $D$ that includes the fact $f$ for every $\hat f\in \hat E$. It is again straightforward that $\sum_{f\in(D\setminus E)w_f}=\sum_{f\in(\hat D\setminus \hat E)w_{\hat f}}$. Now, let $\hat f,\hat g\in \hat E$ such that $\set{\hat f,\hat g}\not\models (X\rightarrow Y)$. We have that $\hat f[A_i]=\hat g[A_i]$ for every $A_i\in X$; thus, $\pi_X\hat f=\pi_X\hat g$ and from the construction of $D$, it holds that $f[A]=g[A]$. On the other hand, the fact that $\hat f[A_i]\neq\hat g[A_i]$ for some $A_i\in Y$ together with the fact that $X\cup Y=X\cup X'=\set{A_1,\dots,A_k}$ imply that $\pi_{X'}\hat f\neq\pi_{X'}\hat g$ and $f[B]\neq g[B]$. Hence, $\set{f,g}\not\models (A\rightarrow B)$. We can similarly prove that if $\set{\hat f,\hat g}\not\models (X'\rightarrow Y')$, then $\set{f,g}\not\models (B\rightarrow A)$, which again implies that $\sum_{\varphi\in\depset}w_\varphi|\vio(E,\varphi)|=\sum_{\hat\varphi\in\hat\depset}w_{\hat\varphi}|\vio(\hat E,\hat\varphi)|$, and the concludes our proof.
    \end{proof}

\end{document}